\newif\ifconf\conffalse
\ifconf
\documentclass[letterpaper]{sig-alternate}
\else
\documentclass[letterpaper,11pt]{article}
\fi
\usepackage{url}
\ifconf
\else
\usepackage{amsthm}
\usepackage{fullpage}
\fi

\usepackage{url}
\usepackage{amssymb,amsmath,amsfonts,cite,hyperref}
\usepackage{graphicx}
\usepackage[usenames]{color}

\DeclareMathOperator{\polylog}{polylog}
\DeclareMathOperator{\argmin}{argmin}
\DeclareMathOperator{\poly}{poly}
\DeclareMathOperator{\nnz}{nnz}

\DeclareMathSymbol{\qedsymb} {\mathord}{AMSa}{"04}

\newcommand{\eps}{\varepsilon}
\renewcommand{\epsilon}{\varepsilon}

\mathchardef\mhyphen="2D

\newcommand{\ceil}[1]{\left\lceil #1 \right\rceil}

\newcommand{\prob}[1]{\operatorname{Pr}\left[\,#1\,\right]}

\newcommand{\var}[1]{\operatorname{Var}\left[\,#1\,\right]}

\newcommand{\R}{\mathbb{R}}

\newcommand{\inprod}[1]{\langle #1 \rangle}

\newcommand{\E}{\mathbb{E}}

\renewcommand{\Pr}{\mathbb{P}}

\newcommand{\EquationName}[1]{\label{eq:#1}}

\newcommand{\LemmaName}[1]{\label{lem:#1}}

\newcommand{\SectionName}[1]{\label{sec:#1}}
\newcommand{\TheoremName}[1]{\label{thm:#1}}

\newcommand{\Equation}[1]{Eq.\:\eqref{eq:#1}}

\newcommand{\Lemma}[1]{Lemma~\ref{lem:#1}}

\newcommand{\Section}[1]{Section~\ref{sec:#1}}
\newcommand{\Theorem}[1]{Theorem~\ref{thm:#1}}

\newtheorem{theorem}{Theorem}
\newtheorem{corollary}[theorem]{Corollary}
\newtheorem{definition}[theorem]{Definition}

\newtheorem{lemma}[theorem]{Lemma}
\newtheorem{remark}[theorem]{Remark}

\newcommand{\proofbelow}{3pt}
\newcommand{\afterproof}{\hfill $\blacksquare$ \par \vspace{\proofbelow}}
\newcommand{\aftersubproof}{\hfill $\Box$ \par \vspace{\proofbelow}}
\renewenvironment{proof}{\noindent\textbf{Proof.}\,}{\afterproof}

\renewcommand{\th}{\ifmmode{^{\textrm{th}}}\else{\textsuperscript{th}\ }\fi}

\newcommand{\comment}[1]{}

\begin{document}

\author{Jelani Nelson\thanks{Institute for Advanced
    Study. \texttt{minilek@ias.edu}. Supported by NSF
  CCF-0832797 and NSF DMS-1128155.}
\and Huy L. Nguy$\tilde{\hat{\mbox{e}}}$n\thanks{Princeton
  University. \texttt{hlnguyen@princeton.edu}. Supported in part by
  NSF CCF-0832797 and a Gordon Wu fellowship.}}

\title{Sparsity Lower Bounds for Dimensionality Reducing Maps}

\maketitle

\begin{abstract}
We give near-tight lower bounds for the sparsity required in several
dimensionality reducing linear maps. First, consider the
Johnson-Lindenstrauss (JL) lemma which states that for any set of $n$
vectors in $\R^d$ there is a matrix $A\in\R^{m\times d}$ with $m =
O(\eps^{-2}\log n)$ such that mapping by $A$ preserves pairwise
Euclidean distances of these $n$ vectors up to a $1\pm\eps$
factor. We show that there exists a set of $n$ vectors such that any
such matrix $A$ with at most $s$ non-zero entries per column must have
$s = \Omega(\eps^{-1}\log n/\log(1/\eps))$ as long as $m <
O(n/\log(1/\eps))$. This bound
improves the lower bound of $\Omega(\min\{\eps^{-2},
\eps^{-1}\sqrt{\log_m d}\})$ by [Dasgupta-Kumar-Sarl\'{o}s, STOC
2010], which only held against the stronger property of distributional
JL, and only against a certain restricted class of distributions.
Meanwhile our lower bound is against the JL lemma itself, with no
restrictions.
Our lower bound matches the sparse Johnson-Lindenstrauss 
upper bound of [Kane-Nelson, SODA 2012] up to an $O(\log(1/\eps))$
factor.

Next, we show that any $m \times n$ matrix with the $k$-restricted
isometry property (RIP) with constant distortion must have at least
$\Omega(k\log(n/k))$ non-zeroes per column if $m=O(k\log (n/k))$, the
optimal number of rows of RIP matrices, and $k < n/\polylog n$. This
improves the previous lower bound of $\Omega(\min\{k, n/m\})$ by [Chandar,
2010] and shows that for virtually all $k$ it is impossible to have a
sparse RIP matrix with an optimal number of rows.

Both lower bounds above also offer a tradeoff between sparsity and
the number of rows.

Lastly, we show that any oblivious distribution over subspace
embedding matrices with $1$ non-zero per column and preserving
distances in a $d$ dimensional-subspace up to a constant factor must
have at least $\Omega(d^2)$ rows. This matches one of the 
upper bounds in [Nelson-Nguy$\tilde{\hat{\textnormal{e}}}$n, 2012]
and shows the impossibility of obtaining
the best of both of constructions in that work, namely 1 non-zero
per column and $\tilde{O}(d)$ rows.
\end{abstract}

\section{Introduction}
The last decade has witnessed a burgeoning interest in algorithms for
large-scale
data. A common feature in many of these works is the exploitation of
data sparsity to achieve algorithmic efficiency, for example to have
running times
proportional to the actual complexity of the data rather
than the dimension of the ambient space
it lives in. This approach has found applications in compressed
sensing~\cite{CT05,don06}, dimension
reduction~\cite{BOR10,DKS10,KN10,KN12,WDLSA09}, and
numerical linear algebra~\cite{CW12,MM12,MP12,NN12a}.  Given the success of
these algorithms, it is important to understand their limitations. 
Until now, for most of
these problems it is not known how far one can reduce the running
time on sparse inputs. In this work we make a step towards
understanding the performance of algorithms for sparse data and show
several tight lower bounds.

In this work we provide three main contributions. We give near-optimal
or optimal sparsity lower bounds for Johnson-Lindenstrauss transforms,
matrices satisfying the restricted isometry property for use in
compressed sensing, and subspace embeddings used in numerical linear
algebra. These three contributions are discussed
in \Section{jl-intro}, \Section{ripintro}, and \Section{oseintro},
respectively.

\subsection{Johnson-Lindenstrauss}\SectionName{jl-intro}
The following lemma, due to Johnson and Lindenstrauss~\cite{JL84}, has
been used widely in many areas of computer science to reduce data
dimension.
\begin{theorem}[Johnson-Lindenstrauss (JL) lemma {\cite{JL84}}]
For any $0<\eps<1/2$ and any $x_1,\ldots,x_n$ in $\R^d$, there exists
$A\in\R^{m\times d}$ with $m = O(\eps^{-2}\log n)$ such that for all
$i,j\in[n]$\footnote{Here and throughout this paper, $[n]$ denotes the
set $\{1,\ldots,n\}$.}, $$ \|Ax_i - Ax_j\|_2 = (1\pm\eps)\|x_i - x_j\|_2 .$$
\end{theorem}
Typically one uses the lemma in algorithm design by mapping some
instance of
a high-dimensional computational geometry problem to a lower
dimension. The running time to solve the instance then becomes the
time needed for the lower-dimensional problem, plus the time 
to perform the matrix-vector
multiplications $Ax_i$; see \cite{Indyk01,Vempala04} for further
discussion. This latter step
highlights the importance of
having a JL matrix  supporting fast matrix-vector multiplication.
The original proofs of the JL lemma took $A$ to be a random dense
matrix, e.g.\ with i.i.d.\ Gaussian, Rademacher, or even subgaussian
entries \cite{Achlioptas03,AV06,DG03,FM88,IM98,JL84,Matousek08}. The
time to compute $Ax$ then becomes $O(m\cdot \|x\|_0)$, where $x$ has
$\|x\|_0\le d$ non-zero entries.

A beautiful work of Ailon and Chazelle \cite{AC09} described a
construction of a JL matrix $A$ supporting matrix-vector
multiplication in time $O(d\log d + m^3)$, also with $m =
O(\eps^{-2}\log n)$. This was improved to $O(d\log d + m^{2+\gamma})$
\cite{AL09} with the same $m$ for any constant $\gamma>0$, or to
$O(d\log d)$ with $m = O(\eps^{-2}\log n\log^4 d)$
\cite{AL11,KW11}. Thus if $\eps^{-2}\log n \ll \sqrt{d}$ one can
obtain nearly-linear $O(d\log d)$ embedding time with the same target
dimension $m$ as the original JL lemma, or one can also obtain
nearly-linear time for any setting of $\eps,n$ by increasing $m$
slightly by $\polylog d$ factors.

While the previous paragraph may seem to present the end of the story,
in fact note that the ``nearly-linear'' $O(d\log d)$ embedding time is
actually much
worse than the original $O(m\cdot \|x\|_0)$ time of dense JL
matrices when $\|x\|_0$ is very small, i.e.\ when $x$ is
sparse. Indeed, in several applications we expect $x$ to be
sparse. 
Consider the bag of words model in information
retrieval: in for example an email spam collaborative filtering system
for {Yahoo!~Mail} \cite{WDLSA09}, each email is treated as a
  $d$-dimensional
vector where $d$ is the size of the lexicon. The $i$th entry of the
vector is some
weighted count of the number of occurrences of word $i$ (frequent
words like ``the'' should be weighted less heavily). A machine
learning algorithm is employed to learn a spam classifier, which
involves dot products of email vectors with some learned classifier
vector, and JL dimensionality reduction is used to speed up the
repeated dot products that are computed during training. Note that in
this scenario we expect $x$ to be sparse since most emails do not
contain nearly every word in the lexicon. 
An even starker scenario is the turnstile streaming model, where the
vectors $x$ may receive coordinate-wise updates in a
data stream. In this case maintaining $Ax$ in a stream given some
update of the form ``add $v$ to $x_i$'' requires adding $vAe_i$ to
the compression $Ax$ stored in memory. Since $\|e_i\| = 1$, we would not
like to spend $O(d\log d)$ per streaming update.

The intuition behind all the works \cite{AC09,AL09,AL11,KW11} to
obtain $O(d\log d)$ embedding time was as follows. Picking $A$ to be a
scaled sampling matrix (where each row has a $1$ in a random location)
gives the correct expectation for $\|Ax\|_2^2$, but the variance
may be too high. Indeed, the variance is high exactly when $x$ is
sparse; consider the extreme case where $\|x\|_0 =1$ so that sampling
is not even expected to see the non-zero coordinate unless $m \ge
d$. These works then all essentially proceed by randomly
preconditioning $x$ to ensure that $x$ is very well-spread (i.e.\ far
from sparse) with high probability, so that sampling works, and thus
fundamentally cannot take advantage of input sparsity.
One way of obtaining faster matrix-vector multiplication for sparse
inputs is to have sparse JL matrices $A$. Indeed, if $A$ has at most
$s$ non-zero entries per column then $Ax$ can be computed in
$O(s\cdot \|x\|_0 + m)$ time. A line of work
\cite{Achlioptas03,Matousek08,DKS10,BOR10,KN10,KN12} investigated the
value $s$ achievable in
a JL matrix, culminating in \cite{KN12}
showing that it is possible to simultaneously have $m =
O(\eps^{-2}\log
n)$ and $s = O(\eps^{-1}\log n)$. Such a sparse JL
transform thus speeds up embeddings by a factor of
roughly $1/\eps$ without increasing the target dimension.

\paragraph{Our Contribution I:} We show that for any $n\ge 2$ and any
$\eps = \Omega(1/\sqrt{n})$, there exists a set of $n$ vectors
$x_1,\ldots,x_n\in\R^n$ such that any JL matrix for this set of
vectors with $m = O(\eps^{-2}\log n)$ rows requires column sparsity $s
= \Omega(\eps^{-1}\log n/\log(1/\eps))$ as long as $m = O(n/\log
(1/\eps))$. Thus the sparse JL transforms
of \cite{KN12} achieve optimal sparsity up to an $O(\log(1/\eps))$
factor.
In fact this lower bound on $s$ continues to hold even if $m =
O(\eps^{-c}\log n)$
for any positive constant $c$.

\bigskip

Note that if $m = n$ one can simply take $A$ to be the identity
matrix which achieves $s=1$, and thus the restriction $m =
O(n/\log(1/\eps))$ is nearly optimal. Also note that 
we can assume $\eps =
\Omega(1/\sqrt{n})$ since otherwise $m = \Omega(n)$ is required
 in any JL matrix \cite{Alon09}, and thus the $m = O(n/\log(1/\eps))$
 restriction is no worse than requiring $m = O(n/\log n)$.
Furthermore if all the entries of $A$ are required to be equal in
magnitude, our lower bound holds as long as $m\le n/10$.

Before our work, only a restricted lower bound of $s = \Omega(\min\{1/\eps^2,
\eps^{-1}\sqrt{\log_m d}\})$ had been shown~\cite{DKS10}. In fact this
lower bound only applied to the {\it distributional JL problem}, a
much stronger guarantee where one wants to design a distribution over
$m\times d$ matrices such that any fixed vector $x$ has $\|Ax\|_2
= (1\pm \eps)\|x\|_2$ with probability $1-\delta$ over the choice of
$A$. Indeed any distributional JL construction yields the JL lemma by
setting $\delta = 1/n^2$ and union bounding over all the $x_i - x_j$
difference vectors. Thus, aside from the weaker lower bound on $s$,
\cite{DKS10} only provided a lower bound against this stronger
guarantee, and furthermore only for a certain restricted class of
distributions that made certain independence assumptions amongst
matrix entries, and also assumed certain bounds on the sum of fourth
moments of matrix entries in each row.

It was shown by Alon \cite{Alon09} that $m = \Omega(\eps^{-2}\log n /
\log(1/\eps))$
is required for the set of points $\{0,e_1,\ldots,e_n\}$ and $d = n$
as long as $1/\eps^2 < n/2$. Here $e_i$ is the $i$th standard basis
vector. Simple manipulations show that, when appropriately scaled, any JL
matrix $A$ for this set of vectors  is {\it
  $O(\eps)$-incoherent}, in the sense
that all its columns $v_1,\ldots,v_n$ have unit $\ell_2$ norm and the
dot products $\inprod{v_i, v_j}$
between pairs of columns are all at most $O(\eps)$ in magnitude.
We study this exact same hard input to the JL lemma;
what we show is that any such matrix $A$ must have column sparsity $s
= \Omega(\eps^{-1}\log n/\log(1/\eps))$. 

In some sense our lower bound
can be viewed as a generalization of the Singleton bound for
error-correcting codes in a certain parameter regime. The Singleton
bound states that for any set of $n$ codewords
with block length $t$, alphabet size $q$, and relative distance $r$,
it must be that $n \le q^{t-r+1}$. If the code has relative distance
$1-\eps$ then $t-r \le \eps t$, so that if $t\ge 1/\eps$ the Singleton
bound implies $t = \Omega(\eps^{-1}\log n/\log q)$. The connection to
incoherent matrices (and thus the JL lemma), observed in
\cite{Alon09}, is the following. For any such code
$\{C_1,\ldots,C_n\}$, form a matrix $A\in\R^{m\times n}$ with $m =
qt$. The rows are partitioned into $t$ chunks each of size $q$. In the
$i$th column of $A$, in the $j$th chunk we put a $1/\sqrt{t}$ in the
row of
that chunk corresponding to the symbol $(C_i)_j$, and we put zeroes
everywhere else in that column. All columns then have $\ell_2$ norm
$1$, and the code having relative distance $1-\eps$ implies that all
pairs of columns have dot products at most $\eps$. The Singleton bound
thus implies that any incoherent matrix formed from codes in this way
has $t = \Omega(\eps^{-1}\log n/\log q)$. Note the column
sparsity of $A$ is $t$, and thus this matches our lower bound for
$q \le
\poly(1/\eps)$. Our sparsity lower bound
thus recovers this Singleton-like bound, without the requirement that
the matrix takes this special structure of being formed from a code in
the manner described above. One reason this is perhaps surprising is
that incoherent matrices from codes have all nonnegative entries; our
lower bound thus implies that the use of negative entries cannot be
exploited to obtain sparser incoherent matrices.

\subsection{Compressed sensing and the restricted isometry property}\SectionName{ripintro}
Another object of interest are matrices satisfying the restricted
isometry property (RIP). Such matrices are widely used in compressed
sensing.

\begin{definition}[\cite{CT05,crt06,Candes08}]
For any integer $k > 0$, a matrix $A$ is said to have the
$k$-restricted isometry property with distortion $\delta_k$ if
$(1-\delta_k)\|x\|_2^2 \le \|Ax\|_2^2 \le (1+\delta_k) \|x\|_2^2$ for
all $x$ with $\|x\|_0 \le k$.
\end{definition}

The goal of the area of compressed sensing is to take few nonadaptive
linear
measurements of a vector $x\in\R^n$ to allow for later recovery from
those measurements. That is to say, if those measurements are
organized as the rows of some matrix $A\in\R^{m\times n}$, we would
like to recover $x$ from $Ax$. Furthermore, we would like do so with $m
\ll n$ so that $Ax$ is a {\it compressed} representation of $x$. Of
course if $m<n$ we cannot recover all vectors $x\in\R^n$
with any meaningful guarantee, since then $A$ will have a non-trivial
kernel, and $x,x+y$ are indistinguishable for
$y\in\ker(A)$. Compressed sensing literature has typically focused on
the case of $x$ being sparse \cite{crt06b,don06}, in which case 
a recovery algorithm could hope to recover $x$ by finding
the
sparsest $\tilde{x}$ such that $A\tilde{x} = Ax$.

The works \cite{Candes08,crt06,CT05} show that if $A$ satisfies
the $2k$-RIP with distortion $\delta_k < \sqrt{2} - 1$, and if $x$ is
$k$-sparse, then given $Ax$ there is a polynomial-time solvable linear
program to recover $x$. In fact for any $x$, not
necessarily sparse, the linear program recovers a vector
$\tilde{x}$ satisfying
$$ \|x - \tilde{x}\|_2 \le O(1/\sqrt{k}) \cdot
\inf_{\|z\|_0 \le k} \|x - z\|_1 ,$$
known as the {\it $\ell_2/\ell_1$ guarantee}.
That is, the recovery error depends on the $\ell_1$ norm of the
best $k$-sparse approximation $z$ to $x$.

It is known \cite{DIPW10,GG84,Kashin77} that any matrix $A$ allowing
for the $\ell_2/\ell_1$ guarantee simultaneously for all
vectors $x$, and thus RIP matrices, must have $m = \Omega(k\log(n/k))$
rows. For completeness we give a proof of the new stronger lower
bound $m = \Omega(\log^{-1}(1/\delta_k)(\delta_k^{-1}k\log(n/k) +
\delta_k^{-2}k))$ in \Section{rip-rows},
though we remark here that current uses of RIP all take $\delta_k =
\Theta(1)$.

Although the recovery $\tilde{x}$ of $x$ can be found in polynomial
time as mentioned above, this polynomial is quite large
as the algorithm involves solving a linear program with $n$ variables and $m$
constraints. This downside has led researchers to design alternative
measurement and/or recovery schemes which allow for much faster
sparse recovery, sometimes even at the cost of obtaining a recovery
guarantee weaker than 
$\ell_2/\ell_1$ recovery for the sake of algorithmic performance. Many
of these
schemes are iterative, such as CoSaMP
\cite{NT09}, Expander Matching Pursuit \cite{IR08}, and several others
\cite{BI09,BIR08,BD08,DTDS12,Foucart11,GK09,NV09,NV10,TG07}, and
several of their running times depend on the product of the number of
iterations and the time required to multiply by $A$ or $A^*$ (here
$A^*$ denotes the conjugate transpose of $A$). Several of these
algorithms furthermore apply $A,A^*$ to vectors which are themselves
sparse. Thus, recovery time
is improved significantly in the case that $A$ is sparse. Previously
the only known lower bound for column sparsity $s$ for an RIP matrix
with an optimal $m=\Theta(k\log(n/k))$ number of rows was $s =
\Omega(\min\{k, n/m\})$ \cite{ChandarThesis}. 
Note that if an RIP construction existed matching
the \cite{ChandarThesis} column sparsity lower bound, application to a
$k$-sparse vector would take time $O(\min\{k^2, nk/m\})$, which is
always $o(n)$ and can be very fast for small $k$.
Furthermore, in several applications
of compressed sensing $m$ is very close to $n$, in which case an
$\Omega(n/m)$ lower bound on column sparsity does not rule out very
sparse RIP matrices. For
example, in applications of compressed sensing to magnetic resonance
imaging, \cite{LDP07} recommended setting the number of
measurements $m$
to be between $5\mhyphen 10\%$ of $n$ to obtain good performance for
recovery of brain and angiogram images. 
We remark that one could also obtain speedup
by using structured RIP matrices, such as
those obtained by sampling rows of the discrete Fourier matrix
\cite{CT06}, though such constructions require matrix-vector
multiplication time $\Theta(n\log n)$ independent of input sparsity.

Another upside of sparse RIP matrices is that they allow faster
algorithms for encoding $x\mapsto Ax$. If $A$ has $s$ non-zeroes per
column and $x$ receives, for example, turnstile streaming
updates, then the
compression $Ax$ can be maintained on the fly in $O(s)$ time per
update (assuming the non-zero entries of
any column of $A$ can be recovered in $O(s)$ time).

\paragraph{Our Contribution II:} We show as long as $k <
n/\polylog n$, any $k$-RIP matrix with distortion $O(1)$ and 
$m = \Theta(k\log (n/k))$ rows with $s$ non-zero
entries per column must have $s = \Omega(k\log(n/k))$. That is, RIP
matrices with the optimal number of rows must be dense for almost the
full range of $k$ up to
$n$. This lower bound strongly rules out any hope for faster recovery
and compression algorithms for compressed sensing  by
using sparse RIP matrices as mentioned above.

\bigskip

We note that any sparsity lower bound should fail as $k$ approaches
$n$ since the $n\times n$ identity matrix trivially satisfies $k$-RIP
for any $k$ and has column sparsity $1$. Thus, our lower bound holds
for almost the full range of parameters for $k$.

\subsection{Oblivious Subspace Embeddings}\SectionName{oseintro}

The last problem we consider is the oblivious subspace embedding
(OSE) problem. Here one aims to design a distribution $\mathcal{D}$
over $m\times n$
matrices $A$ such that for any $d$-dimensional subspace $W\subset\R^n$,
$$\Pr_{A\sim \mathcal{D}}(\forall x\in W\ \|A x\|_2 \in (1\pm
\eps)\|x\|_2) > 2/3 .$$
Sarl\'{o}s showed in \cite{Sarlos06} that OSE's are useful for
approximate least squares regression and low rank approximation, and
they have also been shown useful for approximating statistical
leverage scores \cite{DMMW12}, an important concept in statistics and
machine learning. See \cite{CW12} for an overview of several
applications of OSE's.

To give more details of how OSE's are typically used, consider the
example 
of solving an overconstrained least-squares regression problem, where
one must compute $\argmin_x \|Sx - b\|_2$ for some $S\in\R^{n\times
  d}$. By overconstrained we mean $n> d$, and really one should imagine
$n \gg d$ in what follows. There is a closed form solution for the
minimizing vector $x$,
which requires computing the Moore-Penrose pseudoinverse of $S$. The
total running time is $O(nd^{\omega - 1})$, where $\omega$ is the
exponent of square matrix multiplication.

Now suppose we are only interested in finding some $\tilde{x}$ so that
$$ \|S\tilde{x} - b\|_2 \le (1+\eps)\cdot \argmin_x \|Sx - b\|_2 .$$
Then it suffices to have a matrix $A$ such that $\|Az\|_2 = (1\pm
O(\eps))\|z\|_2$ for all $z$ in the subspace spanned by $b$ and the
columns of $A$, in which case we could obtain such an $\tilde{x}$ by
solving the
new least squares regression problem of computing $\argmin_{\tilde{x}}
\|AS\tilde{x} - Ab\|_2$. If $A$ has $m$ rows, the new running time is
the sum of three terms: (1) the time to compute $Ab$, (2) the
time to compute $AS$, and (3) the $O(md^{\omega -1})$ time required to
solve the new least-squares problem. It turns out it is possible to
obtain such an $A$ with $m=O(d/\eps^2)$ by choosing, for example, a
matrix with independent Gaussian entries (see e.g.\
\cite{Gordon88,KM05}), but
then computing $AS$ takes time $\Omega(nd^{\omega -
  1})$, providing no benefit.

The work of Sarl\'{o}s picked $A$ with special structure so that $AS$
can be computed in time $O(nd\log n)$, namely by using the Fast
Johnson-Lindenstrauss Transform of \cite{AC09} (see also
\cite{Tropp11}). Unfortunately the time
is $O(nd\log n)$ even for sparse matrices $S$, and several
applications require solving numerical linear algebra problems on
sparse matrix inputs. For example in the Netflix matrix where rows are
users and columns are movies, and $S_{i,j}$ is some rating score, 
$S$ is very sparse since most users rate only a tiny fraction of all
movies \cite{ZWSP08}. If $\nnz(S)$ denotes the number of non-zero
entries of $S$, we would like running times closer to $O(\nnz(S))$ than
$O(nd\log n)$ to multiply $A$ by $S$. Such a running time would be
possible, for example, if $A$ only had $s = O(1)$ non-zero entries per
column.

In a recent and surprising work, Clarkson and Woodruff \cite{CW12} gave
an OSE with $m = \poly(d/\eps)$ and $s=1$, thus providing fast
numerical linear algebra algorithms for sparse matrices. For example,
the running time for least-squares regression becomes $O(\nnz(A) +
\poly(d/\eps))$. The dependence on $d,\eps$ was improved in
\cite{NN12a} to $m = O(d^2/\eps^2)$. The work \cite{NN12a} also showed
how to obtain $m = O(d^{1+\gamma}/\eps^2)$, $s=O(1/\eps)$ for any
constant $\gamma>0$ (the constant in the big-Oh depends polynomially
on $1/\gamma$), or $m = (d\polylog d)/\eps^2$, $s = (\polylog
d)/\eps$. It is thus natural to ask whether one can obtain the
best of both worlds: can there be an OSE with $m \approx d/\eps^2$ and
$s = 1$?

\paragraph{Our Contribution III:} 
In this work we show
that any OSE such that all matrices in its support have $m$ rows and
$s=1$ non-zero entries per column must have $m = \Omega(d^2)$ if $n
\ge 2d^2$.
Thus for constant $\eps$ and large $n$, the upper bound of
\cite{NN12a} is optimal.

\subsection{Organization}
In \Section{jlproof} we prove our lower bound for the sparsity
required in JL
matrices. In \Section{riplb} we give our sparsity lower bound for RIP
matrices, and in \Section{subspacelb} we give our lower bound on the
number of rows for OSE's having sparsity $1$.
In \Section{rip-rows} we
give a lower bound involving $\delta_k$ on the number of rows in an
RIP matrix, and in \Section{future} we state an open problem.
\section{JL Sparsity Lower Bound}\SectionName{jlproof}
Define an $\eps$-incoherent matrix $A\in\R^{m\times n}$ as any matrix
whose columns have unit $\ell_2$ norm, and such that every pair of
columns has dot product at most $\eps$ in magnitude.
A simple observation of \cite{Alon09} is that any JL matrix $A$ for
the set of
vectors $\{0,e_1,\ldots,e_n\}\in\R^n$, when its columns are scaled by
their $\ell_2$ norms, must be $O(\eps)$-incoherent.

In this section, we consider an $\eps$-incoherent matrix
$A\in\R^{m\times n}$ with at most $s$ non-zero entries per column. We
show a lower bound on $s$ in terms of $\eps,n,m$. In particular if $m
= O(\eps^{-2}\log n)$ is the number of rows guaranteed by the JL
lemma, we show that $s = \Omega(\eps^{-1}\log n/\log(1/\eps))$ as long
as $ m < n/\polylog n$. In fact if all the entries in $A$ are either
$0$ or equal in magnitude, we show that the lower bound even holds up
to $m < n/10$.

In \Section{signjl} we give the lower bound on $s$ in the case that
all entries in $A$ are in 
$\{0,1/\sqrt{s},-1/\sqrt{s}\}$. In \Section{generaljl} we give our
lower bound without making any assumption on the magnitudes of entries
in $A$. Before proceeding further, we prove a couple lemmas used
throughout this section, and also later in this paper. Throughout this
section $A$ is always an $\eps$-incoherent matrix.

\begin{lemma}\LemmaName{heart}
For any $x\ge 2\eps$, $A$ cannot
have any row with at least $5/x$ entries greater than $\sqrt{x}$, nor
can it have any
row with at least $1/x$ entries less than $-\sqrt{x}$.
\end{lemma}
\begin{proof}
For the sake of contradiction, suppose $A$ did have such a row, say
the $j$th row. Suppose $A_{j,i_1},\ldots,A_{j,i_N} > \sqrt{x}$ for
some $x\ge 2\eps$, where $N \ge 5/x$ (the case where they are each
less than $-\sqrt{x}$ is argued identically). Let $v_i$ denote the
$i$th
column of $A$. Let $u_i$ be $v_i$ but with the $j$th
coordinate replaced with $0$. Then for any $k_1,k_2\in[N]$
$$\inprod{u_{i_{k_1}}, u_{i_{k_2}}}
\le \inprod{v_{i_{k_1}}, v_{i_{k_2}}} - x \le \eps - x \le
-x/2.$$

Thus we have
$$ 0\le \left\|\sum_{j=1}^N u_{i_j}\right\|_2^2 \le N - xN(N-1)/4 ,$$
and rearranging gives the contradiction
$ 1/x \ge (N-1)/4 > 1/x $.
\end{proof}

\begin{lemma}\LemmaName{deriv}
Let $s,q,r$ be positive reals with $q/r\ge 2$ and $s \le q/e$. Then if
$s\ln(q/s) \ge r$ it must be the case that $s = \Omega(r/\ln(q/r))$.
\end{lemma}
\begin{proof}
Define the function $f(s) =
s\ln(q/s)$. Then $f'(s) = \ln(q/(es))$ is increasing for $s \le
q/e$. Then since $q/r \ge 2$, for $s =
cr\ln(q/r)$ for constant $c>0$ we have the equality $s\ln(q/s) =
cr/\ln(q/r)\ln((q/r)
\ln(q/r)) = (c+o_{q/r}(1))r\ln(q/r)$, where the $o_{q/r}(1)$ term goes
to zero
as $q/r\rightarrow \infty$. Thus for $c$ sufficiently small we have
that the $c+o_{q/r}(1)$ term must be less than $1$,
so in order to have $f(s) \ge r$, since $f$
is increasing we must have $s = \Omega(r/\ln(q/r))$. 
\end{proof}

\subsection{Sign matrices}\SectionName{signjl}

In this section we consider the case that all entries of $A$ are
either $0$ or $\pm 1/\sqrt{s}$ and show a lower bound on $s$ in this
case.

\begin{lemma}\LemmaName{sign}
Suppose $m < n/10$ and all entries of $A$ are in
$\{0,1/\sqrt{s},-1/\sqrt{s}\}$. Then $s \ge 1/(2\eps)$.
\end{lemma}
\begin{proof}
For the sake of contradiction suppose $s < 1/(2\eps)$.
There are $ns$ non-zero entries in $A$ and thus at least $ns/2$
of these entries have the same sign by the pigeonhole principle; wlog
let us say $1/\sqrt{s}$ appears at least $ns/2$ times. Then again by
pigeonhole some row $j$ of $A$ has $N = ns/(2m)$ values that are
$1/\sqrt{s}$. The claim now follows by \Lemma{heart} with $x =
1/\sqrt{s}$.
\end{proof}

We now show how to improve the bound to the desired form.


\begin{theorem}
Suppose $m < n/10$ and all entries of $A$ are in
$\{0,1/\sqrt{s},-1/\sqrt{s}\}$. Then $s \ge \Omega(\eps^{-1}\log n /
\log(m/\log n))$.
\end{theorem}
\begin{proof}
We know $s \ge 1/(2\eps)$ by \Lemma{sign}. Let $t = 2\eps s \ge 1$.
Every $v_i$ has $\binom{s}{t}$
subsets of size $t$ of non-zero coordinates. Thus by pigeonhole there
exists a set of $t$ rows $i_1,\ldots,i_t$ and $N = n\binom{s}{t}/(2^t
\binom{m}{t})$ columns $v_{j_1},\ldots,v_{j_N}$ such that for each row
all entries in those columns are $1/\sqrt{s}$ in magnitude and have
the same sign (the signs may vary across rows). Letting $u_j$ be $v_j$
but with those $t$ coordinates set to $0$, we have
$$\inprod{u_{j_{k_1}}, u_{j_{k_2}}}
= \inprod{v_{j_{k_1}}, v_{j_{k_2}}} - t/s \le \eps - t/s \le
-t/(2s).$$

Thus we have
$$ 0\le \left\|\sum_{k=1}^N u_{j_k}\right\|_2^2 \le N - tN(N-1)/(4s) $$
so that rearranging gives
$$ s \ge t(N-1)/4 = (t/4)\cdot
\left(\frac{n\binom{s}{t}}{2^t\binom{m}{t}} - 1\right) \ge (t/4)
\cdot (n(s/(2em))^t - 1).$$
Suppose $s < c\eps^{-1}\log n/\log(2em/n)$ for some small constant $c$
so that $n(s/(2em))^t \ge 2$. Then 
$$ s \ge (tn/8)\cdot (s/(2em))^t .$$
Thus
$$ \frac{\eps n}{4} = \frac{tn}{8s} \le \left(\frac{2em}{s}\right)^t
. $$
Taking the natural logarithm of both sides gives
$$ s\ln\left(\frac{2em}{s}\right) \ge
\frac{1}{2\eps}\ln\left(\frac{\eps n}{4}\right) .$$
Define $q = 2em$, $r = \eps^{-1}\ln(\eps n/4)/2$. Then $s \le q/e$,
since $s \le m$. By \cite{Alon09} we must have $m =
\Omega(\eps^{-2}\log n/\log(1/\eps))$, so $q/r \ge 2$ for $\eps$
smaller than some fixed constant. 
Thus by \Lemma{deriv} we have $s = \Omega(r/\ln(q/r))$. The theorem
follows since $\log(\eps m/\log n) = \Theta(m/\log n)$ since $m =
\Omega(\eps^{-2}\log n/\log(1/\eps))$ \cite{Alon09}.
\end{proof}

\begin{corollary}
Suppose $m \le \poly(1/\eps)\cdot \log n < n/10$ and all entries of
$A$ are in
$\{0,1/\sqrt{s},-1/\sqrt{s}\}$. Then $s \ge \Omega(\eps^{-1}\log n /
\log(1/\eps))$.
\end{corollary}

\subsection{General matrices}\SectionName{generaljl}
We now consider arbitrary sparse and nearly orthogonal matrices
$A\in\R^{m\times n}$. That is, we no longer require the non-zero
entries of $A$ to be $1/\sqrt{s}$ in magnitude.

\begin{lemma}\LemmaName{general}
Suppose $m < n/(20\ln(1/2\eps))$. Then $s \ge 1/(4\eps)$.
\end{lemma}
\begin{proof}
For the sake of contradiction suppose $s < 1/(4\eps)$.
We know by \Lemma{heart} that for any $x\ge 2\eps$, no row of $A$ can
have more than $5/x$ entries of value at least $\sqrt{x}$ in magnitude
and of the same
sign. Define $S_i = \{j : A_{i,j}^2 \ge 2\eps\}$.
Let $S_i^+$ be the subset of indices $j$ in $S_i$ with $A_{i,j} > 0$,
and define $S_i^- = S_i\backslash S_i^+$.
Let $X$ denote the square
of a random positive value from $S_i^+$. Then
$$\sum_{j\in S_i^+} A_{i,j}^2 = |S_i^+| \cdot \E X  =
|S_i^+| \cdot \int_{0}^1 \Pr\left(X > x\right)dx
\le 2\eps |S_i^+| + \int_{2\eps}^1 \frac 5x dx = 2\eps |S_i^+| +
5\ln(1/2\eps) .$$
By analogously bounding the sum of squares of entries in $S_i^-$, we
have that the sum of
squares of entries at least $\sqrt{2\eps}$ in magnitude is
never more than $2\eps |S_i| + 10\ln(1/2\eps)$ in the $i$th row of
$A$, for any $i$.
Thus
the total sum of squares of all entries in the matrix less than
$\sqrt{2\eps}$ in magnitude is at most $2\eps (n s - \sum_i
|S_i|)$. Meanwhile the sum of all other entries is at most
$2\eps(\sum_i |S_i|) + 10m\ln(1/2\eps)$. Thus the sum of squares of
all entries in the matrix is
at most $2\eps ns + 10m\ln(1/2\eps) < n/2 + 10m\ln(1/2\eps)$, by our
assumption on $s$. This quantity must be $n$, since every column of
$A$ has unit $\ell_2$ norm.
However
for our stated value of $m$ this is impossible since $10m\ln(1/2\eps)
< n/2$, a contradiction.
\end{proof}

We now show how to obtain the extra factor of $\log n / \log(1/\eps)$ in
the lower bound.

\begin{lemma}\LemmaName{lb-core}
Let $0 < \eps < 1/2$.
Suppose $v_1,\ldots,v_n\in\R^m$ each have $\|v\|_2 = 1$ and
$\|v\|_0 \le s$, and furthermore $|\inprod{v_i, v_j}| \le \eps$ for
$i\neq j$. Then for any $t\in [s]$ with $t/s > C\eps$, we must have $s
\ge t(N - 1)/(2C)$ with
$$ N = \ceil{\frac{n}{2^t\binom{m}{t}\binom{2(s+t)}{t}}},\hspace{.5in} C =
2/(1 - 1/\sqrt{2}) .$$
\end{lemma}
\begin{proof}
We label each vector $v_i$ by its {\em $t$-type}, defined in the
following way. The $t$-type of a vector $v_i$ is the set of locations
of the $t$ largest coordinates in magnitude, as well as the signs of
those coordinates, together with a rounding of those top $t$
coordinates so that their squares round to the nearest integer
multiple of $1/(2s)$. In the rounding, values halfway between two multiples are
rounded arbitrarily; say downward, to be concrete. Note that the
amount of $\ell_2$ mass contained in the top $t$ coordinates of any
$v_i$ after such a rounding is at most $1+t/(2s)$, and thus the number of
roundings possible is at most the number of ways to write a positive
integer in $[2s+t]$ as a sum
of $t$ positive integers, which is $\binom{2s+2t}{t}$. Thus
the total number of possible $t$-types is at most
$2^t\binom{m}{t}\binom{2(s+t)}{t}$ ($\binom{m}{t}$ choices of the
largest $t$ coordinates, $2^t$ choices of their signs, and
$\binom{2(s+t)}{t}$ choices for how they round). Thus by the pigeonhole
principle, there exist $N$ vectors $v_{i_1},\ldots,v_{i_N}$ each with
the same $t$-type such that $N \ge
\ceil{n/(2^t\binom{m}{t}\binom{2(s+t)}{t})}$.

Now for these vectors $v_{i_1},\ldots,v_{i_N}$, let $S\subset[n]$ of
size $t$ be the set of the largest coordinates (in magnitude) in each
$v_{i_j}$. Define
$u_{i_j} = (v_{i_j})_{[n]\backslash S}$; that is, we zero out the
coordinates in $S$. Then for $j\neq k\in[N]$,
\allowdisplaybreaks
\begin{align}
\nonumber \inprod{u_{i_j}, u_{i_k}} &= \inprod{v_{i_j}, v_{i_k}} -
\sum_{r\in S}(v_{i_j})_r(v_{i_k})_r\\
\nonumber {}&\le \eps - \sum_{r\in
  S}(v_{i_j})_r((v_{i_j})_r \pm 1/\sqrt{2s})\\
\nonumber {}&\le \eps - \sum_{r\in
  S}\left((v_{i_j})_r^2 - |(v_{i_j})_r|/\sqrt{2s}\right)\\
\nonumber {}&\le \eps - \|(v_{i_j})_S\|_2^2 + \sqrt{t/(2s)}\cdot
\left\|(v_{i_j})_S\right\|_2\\
{}&\le \eps - \left(1 - \frac{1}{\sqrt{2}}\right)t/s
.\EquationName{dotprod}
\end{align}
The last inequality used that $\|(v_{i_j})_S\|_2\ge \sqrt{t/s}$. Also
we pick $t$ to ensure
$t/s > 2\eps/(1 - 1/\sqrt{2})$ so that
the right hand side of \Equation{dotprod} is less than $-((1 -
1/\sqrt{2})/2)t/s = -Ct/s$. The penultimate inequality follows by
Cauchy-Schwarz.
Thus we have
\begin{align}
\nonumber \left\|\sum_{j=1}^N u_{i_j}\right\|_2^2 &= \sum_{j=1}^N
\|u_{i_j}\|_2^2 + \sum_{j\neq k}\inprod{u_{i_j}, u_{i_k}}\\
{}&\le N - C(t/s) N(N-1)/2 \EquationName{small-ell2}
\end{align}

However we also have $\|\sum_j u_{i_j}\|_2^2 \ge 0$, which 
implies $s \ge C(N-1)t/2$ by rearranging \Equation{small-ell2}.
\end{proof}

\begin{theorem}\TheoremName{main}
There is some fixed $0<\eps_0\le 1/2$ so that the following holds. Let
$1/\sqrt{n} < \eps < \eps_0$.
Suppose $v_1,\ldots,v_n\in\R^m$ each have $\|v\|_2 = 1$ and
$\|v\|_0 \le s$, and furthermore $|\inprod{v_i, v_j}| \le \eps$ for
$i\neq j$. Then $s \ge \Omega(\eps^{-1}\log n/\log(m/\log n))$
 as long as $m < O(n/\ln(1/\eps))$.
\end{theorem}
\begin{proof}
By \Lemma{general}, $4\eps s \ge 1$. Set $t = 7\eps s$ so that
\Lemma{lb-core} applies. Then by \Lemma{lb-core}, as long as
$2^t\binom{m}{t}\binom{2(s+t)}{t} \le n/2$,
\begin{align*}
7\eps n = \frac{tn}{s} &\le 4C \cdot 2^t\binom{m}{t}\binom{2(s+t)}{t}\\
&\le 4C\cdot \left(\frac{8e^2m}{49\eps^2 s}\right)^{7\eps s} ,
\end{align*}
where $C$ is as in \Lemma{lb-core}.
Taking the natural logarithm on both sides,
$$\ln (7\eps n/(4C)) \le (7\eps s)\ln\left(\frac{8e^2 m}{49\eps^2
    s}\right)$$
In other words,
$$s \ge \frac{\ln (7\eps n/(4C))}{7\eps \ln\left(\frac{8e^2 m}{49\eps^2
      s}\right)} .$$
Define $r = \ln(7\eps n/(4C))/(7\eps), q = 8e^2m/(49\eps^2)$. Thus we have
$s\ln(q/s) \ge r$. We have that $s \le q/e$ is
always the case for $\eps < 1/2$ since then 
$q/e \ge m$ and we have that $s \le m$. Also note for $\eps$ smaller
than some constant we have that $q/r > 2$ since $m = \Omega(\log n)$
by \cite{Alon09}. Thus by \Lemma{deriv} we have
$s \ge \Omega(r/\ln(q/r))$. Using that $\ln(\eps n) = \Theta(\log n)$
since
$\eps > 1/\sqrt{n}$, and that $2^t\binom{m}{t}\binom{2(s+t)}{t} \le
(8e^2m/(49\eps^2 s)) \le n/2$ for our setting of $t$ when $s =
o(\eps^{-1}\log n / \log(m/(\eps^{-1}\log n)))$ 
gives $s = \Omega(\eps^{-1}\log n/\log(\eps^{-1}m/\log n))$. Since $m
= \Omega(\eps^{-2}\log n/\log(1/\eps))$ \cite{Alon09}, this is
equivalent to our lower bound in the theorem statement.
\end{proof}

\begin{corollary}
Let $\eps,m,s$ be as in \Theorem{main}. Then $s = \Omega(\eps^{-1}\log
n / \log(1/\eps))$ as long as $m \le \mathrm{poly}(1/\eps)\cdot \log n
 < O(n/\ln(1/\eps))$.
\end{corollary}

\begin{remark}
\textup{
From \Theorem{main}, we can deduce that for constant $\eps$, in order
for the sparsity $s$ to be a constant independent of $n$, it must be
the case that $m = n^{\Omega(1)}$. This fact rules out very sparse
mappings even when we significantly increase the target dimension.
}
\end{remark}

\section{RIP Sparsity Lower Bound}\SectionName{riplb}
Consider a $k$-RIP matrix $A \in \R^{m\times n}$ with distortion
$\delta_k$ where each column has
at most $s$ non-zero entries. We will show for $\delta_k = \Theta(1)$
that $s$ cannot be very
small when $m$ has the optimal number of rows $\Theta(k\log(n/k))$.

\begin{theorem}\TheoremName{rip-sparsity}
Assume $k \ge 2$, $\delta_k < \delta$ for some fixed universal small
constant $\delta>0$, $m < n/(64\log^3 n)$. Then we must have
$s = \Omega(\min\{k\log (n/k)/\log(m/(k\log(n/k))),
m\})$.
\end{theorem}
\begin{proof}
  Assume for the sake of contradiction that $s <
  \min\{k\log(n/k)/(64\log(m/s)), m/64\}$. Consider the $i$th column of
  $A$ for some fixed $i$. By $k$-RIP, the
  $\ell_2$ norm of each column of $A$ is at least $1-\delta_k > 1/2$,
  so the sum of squares of entries greater than $1/(2\sqrt{s})$ in magnitude is at
  least $1/4$.
Therefore, there exists a scale $1\le t\le \log s$ such
  that the number of entries of absolute value greater than or equal
  to $2^{(t-3)/2}/\sqrt{s}$ is at least $2^{-t-1}s/t^2$.  To see this,
  let $|S|$ be the set of rows $j$ such that $|A_{j,i}| \ge
  1/(2\sqrt{s})$. For the sake of contradiction, suppose that every
  scale $1\le t\le \log s$ has strictly fewer than $2^{-t-1}s/t^2$
  values that are
  at least $2^{(t-3)/2}/\sqrt{s}$ in magnitude (note this also implies
  $|S| < s/4$). Let $X$ be the square
  of a random element of $S$. Then
$$\sum_{j\in S}A_{j,i}^2 = |S|\cdot \E X = |S| \cdot \int_{0}^{\infty}
\Pr\left(X > x\right)dx < \frac{1}{16} +
\int_{1/4s}^{\infty}\Pr\left(X > x\right) dx  < \frac 1{16} +
\sum_{t=1}^\infty \frac{2^t}{8s}\cdot \frac{s}{2^{t+1}t^2} < \frac 14
,$$
a contradiction.
Let a pattern
  at scale $t$ be a subset of size $u=\max\{2^{4-t}s/k, 1\}$ of $[m]$
  along with $u$ signs. There are $2^{-t-1} s/t^2 \choose u$ patterns
  $P$ where $A_{v,i}^2 \ge 2^{t-3}/s$ for all $v\in P$ and the signs
  of $A_{v,i}$ match the signs of $P$.

There are $2^u {m\choose u}$ possible patterns at scale $t$. By an
averaging argument, there exists a scale $t$, and a pattern $P$ such
that the number of columns of $A$ with this pattern is at least $z =
n{2^{-t-1}s/t^2 \choose u }/((\log s)2^u {m\choose u})$. Consider
2 cases.

\paragraph{Case 1 ($z\ge k$):}Pick an arbitrary set of $k$ such
columns. Consider the vector $v$ with $k$ ones at locations
corresponding to those columns and zeroes everywhere else. We have
$\|v\|_2^2 = k$ and for each $j\in P$, we have
$$(Av)_j^2 \ge k^2 2^{t-3}/s .$$
Thus, 
$$\|Av\|_2^2 \ge u k^2 2^{t-3}/s \ge  2k .$$
This contradicts the assumption that $\|Av\|_2^2 \le
(1+\delta_k)\|v\|_2^2$.

\paragraph{Case 2 ($z < k$):} Consider the vector $v$ with $z$ ones at
locations corresponding to those columns and zeroes everywhere else. We
have $\|v\|_2^2 = z$ and for each $j\in P$, we have $(Av)_j^2 \ge z^2
2^{t-3}/s$. Consider 2 subcases.

\paragraph{Case 2.1 ($u=1$):} Then $z = \frac{n2^{-t-2}s/t^2}{(\log
  s)m}$, so
\begin{equation}
\|Av\|_2^2 \ge z^2 2^{t-3}/s \ge \frac{2^{-5}n/t^2}{(\log s)m} \cdot
z \ge 2z .\EquationName{small-m}
\end{equation}

This contradicts the assumption that $\|Av\|_2^2 \le
(1+\delta_k)\|v\|_2^2$.

\paragraph{Case 2.2 ($u = 2^{4-t}s/k$):} . We have

\allowdisplaybreaks
\begin{align}
\nonumber z =&\frac{n{2^{-t-1}s/t^2 \choose u }}{(\log s)2^u {m\choose u}}\\
\nonumber \ge& \frac{n}{\log s}\left(\frac{s}{t^22^{t+2}em}\right)^u\\
\nonumber \ge& \frac{n}{\log s}2^{-(\log(m/s) + \log e + t+2 + 2 \log t)2^{4-t}s/k}\\
\ge& \frac{n}{\log s}2^{-(\log(m/s) + \log e + t+2+ 2 \log t)2^{4-t}\cdot \log(n/k)/(64\log(m/s))}\EquationName{subs}\\
\ge& \frac{n}{\log s}(k/n)^{1/4} \EquationName{mint}\\
\ge& k .\EquationName{small-k}
\end{align}

\Equation{subs} follows from $s < k\log(n/k)/(64\log(m/s))$. \Equation{mint} follows from the fact that $f(t) = (\log(m/s) + \log e + t + 2 + 2\log t)2^{-t}$ is monotonically decreasing for $t\ge 1$. Indeed, 
\begin{align*}
f'(t) &= 2^{-t}\left(-\ln2(\log(m/s) + \log e + 2+t+2\log t\right) +
\frac{2}{t\ln 2} + 1)\\
{}&\le 2^{-t}\left(-9\ln 2 - t\ln 2 +\frac{2}{t\ln 2} + 1\right)\\
{}& \le  0 .
\end{align*}
\Equation{small-k} follows since $k < n/\log^{4/3} n < n/\log^{4/3}
s$, which holds since $k \le m\le n/(64\log^3 n)$.
This contradicts the assumption of Case 2 that $z < k$.

Thus we have $s \ge
  \min\{k\log(n/k)/(64\log(m/s)), m/64\}$ as desired. If $s \ge m/64$
  we are done. Otherwise we have $s \ge
  k\log(n/k)/(64\log(m/s))$. Define $q = m$, $r = k\log(n/k)/64$. Thus
  we have $s\log(q/s) \ge r$. We have $q/r \ge 2$ for $\delta_k$
  smaller than some constant by \Theorem{rip-rows}, and we have $s <
  q/e = m/e$ since we assume we are in the case $s < m/64$. Thus by
  \Lemma{deriv} we have $s = \Omega(r/\ln(q/r))$, which completes the
  proof of the theorem.
\end{proof}

\begin{corollary}
  When $k \ge 2$, $\delta_k < \delta$ for some universal constant
  $\delta>0$, and the number of rows $m =
  \Theta(k\log(n/k)) < n/(32\log^3 n)$, we must have $s = \Omega(k\log
  (n/k))$.
\end{corollary}

\begin{remark}
\textup{
The restriction $m = O(n/\log^3 n)$ in \Theorem{rip-sparsity} was
relevant in \Equation{small-m}. Note the choice of $t^2$ in the
proof was just so that $\sum_t 1/t^2$ converges. We could instead have
chosen $t^{1+\gamma}$ and obtained a qualitatively similar result, but
with the slightly milder restriction $m = O(n/\log^{2+\gamma} n)$,
where $\gamma>0$ can be chosen as an arbitrary constant.
}
\end{remark}
\section{Oblivious Subspace Embedding Sparsity Lower Bound}\SectionName{subspacelb}
In this section, we show a lower bound on the dimension of very sparse
OSE's.

\begin{theorem}
Consider $d$ at least a large enough constant and $n\ge 2d^2$. Any
OSE with matrices $A$ in its support having  $m$ rows and at
most $1$ non-zero entry per column such that with probability at least
$1/5$, the lengths of all vectors in a fixed subspace of dimension $d$
of $\R^n$ are preserved up to a factor $2$, must have $m \ge d^2 /
214$.
\end{theorem}

\begin{proof}
Assume for the sake of contradiction that $m < d^2/214$. By Yao's
minimax principle, we only need to show there exists a distribution
over subspaces such that any fixed matrix $A$ with column sparsity $1$
and too few rows would fail to preserve
lengths of vectors in the subspace with probability more than $4/5$. 

Consider the uniform distribution over subspaces spanned by $d$
standard basis vectors in $\R^n$: $e_{i_1}, e_{i_2}, \ldots, e_{i_d}$
with $i_1,\ldots, i_d\in \{1, \ldots, n\}$. Let $a(i)$ be the row of
the non-zero entry in column $i$ of $A$ and $b(j)$ be the number of
non-zeroes in row $j$. We say $i$ {\em collides with $j$} if $a(i) =
a(j)$. Let the set of {\em heavy} rows be the set of rows $j$ such
that $b(j) \ge \frac{n}{10m}$.

If we pick $i_1, \ldots, i_d$ one by one. Conditioned on $i_1, \ldots,
i_{t-1}$, the probability that $a(i_t)$ is heavy is at least
$\frac{9}{10}-\frac{d}{n}\ge \frac{4}{5}$. Therefore, by a Chernoff
bound, with probability at least $9/10$, the number of indices $i_t$
such that $a(i_t)$ are heavy is at least $3d/4$.

We will show that conditioned on the number of such $i_t$ being at
least $3d/4$, with probability at least $9/10$, two such indices
collide. Let $j_{1}, \ldots, j_{3d/4}$ be indices with $b(a(j_t)) \ge
\frac{n}{10m}$. Conditioned on $a(j_1), \ldots, a(j_{t-1})$, the
probability that $j_t$ does not collide with any previous index is at
most 
$$1-\sum_{u=1}^{t-1} b(a(j_u))/(n-t+1) + (t-1)/(n-t+1)\le e^{-\sum_{u=1}^{t-1} b(a(j_u))/n+2(t-1)/n}
\le e^{-(t-1)/(10m)+2(t-1)/n} .$$

Thus, the probability that no collision occurs is at most
$e^{(-(3d/4)^2 /(40m))+((3d/4)^2/n)} < 1/10$. In other words, collision occurs with
probability at least $9/10$. When collision occurs, the number of
non-zero entries of $AM$, where $M$ is the matrix whose columns are
$e_{i_1}, \ldots, e_{i_d}$, is at most $d-1$ so it has rank at most
$d-1$. Therefore, with probability at least $4/5$, $A$ maps some
non-zero vector in the subspace to the zero vector (any vector $Mx$
for $x\in \ker(AM)$) and fails to preserve the length of all vectors
in the subspace.
\end{proof}

\section{Lower Bound on Number of Rows for RIP Matrices}\SectionName{rip-rows}
In this section we show a lower bound on the number of rows of any
$k$-RIP matrix with distortion $\delta_k$. First we need the following
form of the Chernoff bound.

\begin{theorem}[Chernoff bound]
Let $X_1,\ldots, X_n$ be independent random variables each at most $K$
in magnitude almost surely, and with $\sum_{i=1}^n\E X_i = \mu$
and $\var{\sum_{i=1}^n X_i} = \sigma^2$. Then 
$$ \forall \ \lambda >0,\ \prob{\left|\sum_{i=1}^n X_i - \mu\right| >
  \lambda \sigma} < C\cdot \max\left\{e^{-c\lambda^2}, (\lambda
  K/\sigma)^{-c\lambda\sigma/K}\right\} $$
for some absolute constants $c,C>0$.
\end{theorem}

This form of the Chernoff bound can then be used to show the existence
of a large error-correcting code with high relative distance.

\begin{lemma}\LemmaName{codeexists}
For any $0<\eps\le 1/2$ and integers $k,n$ with $1 \le k \le \eps n/2$,
there exists a $q$-ary code with $q = n/k$ and block length $k$ of
relative distance $1-\eps$, and with size at least
$$ \min\left\{e^{C' \eps^2 n}, e^{C'\eps k\log\left(\frac{\eps
        n}{2k}\right)}\right\} $$
for some absolute constant $C'>0$.
\end{lemma}
\begin{proof}
We take a random code. That is, pick
$$ N = \min\left\{e^{C \eps^2 n}, e^{C\eps k\log\left(\frac{\eps
        n}{2k}\right)}\right\} $$
codewords with alphabet size $q = n/k$ and block length $k$, with
replacement. Now, look at two of these randomly chosen codewords. For
$i = 1,\ldots,k$, let $X_i$ be an indicator random variable for the
event that the $i$th symbol is equal in the two codewords. Then $X =
\sum_{i=1}^k X_i$ is the number of positions at which these two
codewords agree, and $\E X = k^2/n \le \eps k/2$ and $\var{X}
\le k^2/n$. Thus by the Chernoff bound,
$$\Pr\left(|X| > \eps k\right) < C \cdot \max\left\{e^{-c \eps^2 n},
  e^{-c\eps k\log\left(\frac{\eps n}{2k}\right)}\right\} .$$
Therefore by a union bound, a random multiset of $N$ codewords has
relative distance $1-\eps$ with positive probability (in which case it
must also clearly be not just a multiset, but a set).
\end{proof}

Before proving the main theorem of this section, we also need the
following theorem of Alon \cite{Alon09}.

\begin{theorem}[Alon {\cite{Alon09}}]\TheoremName{alon}
Let $x_1,\ldots,x_N\in\R^n$ be such that $\|x_i\|_2 = 1$ for all $i$,
and $|\inprod{x_i,x_j}| \le \eps$ for all $i\neq j$, where $1/\sqrt{n}
< \eps < 1/2$. Then $n = \Omega(\eps^{-2}\log N/\log(1/\eps))$.
\end{theorem}

\begin{theorem}\TheoremName{rip-rows}
For any $0<\delta_k\le 1/2$ and integers $k,n$ with $1 \le k \le
\delta_k n/2$,
any $k$-RIP matrix with distortion $\delta_k$ must have
$\Omega\left(\min\{n/\log(1/\delta_k),
  (k/(\delta_k\log(1/\delta_k)))\log(n / k)\}\right)$ rows.
\end{theorem}
\begin{proof}
Let $C_1,\ldots,C_N$ be a code as in \Lemma{codeexists} with block
length $n/(k/2)$ and alphabet size $k/2$ with 
$$ N \ge \min\left\{e^{C \delta_k^2 n}, e^{C\delta_k k\log\left(\frac{\delta_k
        n}{k}\right)}\right\} .$$
Consider a set of vectors $y_1,\ldots,y_N$ in $\R^n$ defined as
follows. For $j=0,\ldots,k/2-1$, we define $(y_i)_{2jn/k + (C_i)_j} =
\sqrt{2/k}$, and all other coordinates of $y_i$ are $0$. Then we have
$\forall i\ \|y_i\|_2 = 1$, and also $0\le \inprod{y_i,y_j} \le \delta_k$
for all $i\neq j$, and thus $2 - 2\delta_k\le \|y_i - y_j\|_2^2 \le
2$. Since $y_i$ is $k/2$-sparse and $y_i - y_j$ is $k$-sparse for all
$i,j$, we have for any $k$-RIP matrix $A$ with distortion $\delta_k$
$$ \forall i\ \|Ay_i\|_2 = 1\pm\delta_k,\hspace{.5in} \forall i\neq j\
\|Ay_i - Ay_j\|_2^2 = (1\pm \delta_k)^2\cdot (2 \pm 2\delta_k) = 2\pm
9\delta_k .$$

Thus if we define $x_1,\ldots, x_N$ by $x_i = Ay_i /
\|Ay_i\|_2$, then the $x_i$ satisfy the requirements of \Theorem{alon}
with inner products at most $O(\delta_k)$ in magnitude. The lower bound on
the number of rows of $A$ then follows.
\end{proof}

It is also possible to obtain a lower bound on the number of rows of
$A$ in \Theorem{rip-rows} of the form
$\Omega(\delta_k^{-2}k/\log(1/\delta_k))$. This is because a theorem of
\cite{KW11} shows that any such RIP matrix with $k = \Theta(\log n)$, when
its column signs are flipped randomly, is a JL matrix for any set of
$n$ points with high probability. We then know from \Theorem{alon}
that a JL matrix must have $m = \Omega(\delta_k^{-2}\log n/\log(1/\delta_k))$
rows, which is $\Omega(\delta_k^{-2}k/\log(1/\delta_k))$.

\begin{corollary}
Suppose $1/\sqrt{n} \le \delta_k \le 1/2$ and $A\in\R^{m\times n}$ is a
$k$-RIP matrix with
distortion $\delta_k$. Then $m =
\Omega(\log^{-1}(1/\delta_k)\cdot \min\{k\log(n/k)/\delta_k + k/\delta_k^2,
n\})$.
\end{corollary}
\section{Future Directions}\SectionName{future}

For several applications the JL lemma is used as a black box to obtain
dimensionality-reducing linear maps for other problems. For example,
applying the JL lemma with distortion $O(\delta_k)$ on a certain
net with
$N = O\binom{n}{k}\cdot O(1/\delta_k)^k$ vectors yields a $k$-RIP matrix
with distortion $\delta_k$ \cite{BDDW08}. Note in this case, for
constant $\delta_k$, the number of rows one obtains is the optimal
$\Theta(\log N) = \Theta(k\log(n/k))$. Applying the distributional
JL lemma with
distortion $O(\eps)$ to a certain net of size $2^{O(d)}$ yields an OSE
with $m = O(d/\eps^2)$ rows to preserve $d$-dimensional subspaces (see \cite[Fact
10]{CW12}, based on \cite{AHK06}).

Applying the JL lemma in this black-box way using the sparse JL
matrices of \cite{KN12}  yields a factor-$\eps$ improvement 
in sparsity over
using a random dense JL construction, with for example random Gaussian
entries. However, some examples have shown that it is possible to
do much
better by not using the JL lemma statement as a black box, but rather
by analyzing the
sparsity required from the constructions in \cite{KN12} ``from
scratch'' for the problem at hand. For example, the work \cite{NN12a}
showed that one can
have column sparsity $O(1/\eps)$ with $m = O(d^{1+\gamma}/\eps^2)$
rows in an OSE for any $\gamma>0$, which is much better than the
column sparsity $O(d/\eps)$
that is obtained by using the sparse JL theorem as a black box.

We thus pose the following open problem in the realm of understanding
sparse embedding matrices better. Let $\mathcal{D}$ be an OSNAP
distribution \cite{NN12a} over $\mathrm{R}^{m\times n}$ with column
sparsity $s$. The class of OSNAP distributions includes both of the
sparse JL distributions in \cite{KN12}, and more generally an OSNAP
distribution is
 characterized by the following three properties where $A$ is a random
 matrix drawn from $\mathcal{D}$:
\begin{itemize}
\item All entries of $A$ are in $\{0,1/\sqrt{s},-1/\sqrt{s}\}$. We
  write $A_{i,j} = \delta_{i,j}\sigma_{i,j}/\sqrt{s}$ where
  $\delta_{i,j}$ is an indicator random variable for the event
  $A_{i,j} \neq 0$, and the $\sigma_{i,j}$ are independent uniform
  $\pm 1$ r.v.'s.
\item For any $j\in [n]$, $\sum_{i=1}^m \delta_{i,j} = s$ with
  probability $1$.
\item For any $S\subseteq [m]\times [n]$, $\E \prod_{(i,j) \in S}
  \delta_{i,j} \le (s/m)^{|S|}$.
\end{itemize}
Given a set of vectors $V\subset \R^n$,
what is the tradeoff between the number of rows $m$ and the column
sparsity $s$ required for a random matrix $A$ drawn from an OSNAP
distribution to preserve all $\ell_2$ norms of vectors $v\in V$ up to
$1\pm\eps$ simultaneously, with positive probability, as a function of
the geometry of $V$? We are motivated to ask this question by a result
of \cite{KM05}, which states that for a set of vectors
$V\subseteq\R^n$ all of unit $\ell_2$ norm, a matrix with random
subgaussian entries preserves all $\ell_2$ norms of vectors in $V$ up
to $1\pm\eps$ as long as the number of rows $m$ satisfies
\begin{equation}
m \ge C\eps^{-2} \cdot \left(\E_g \sup_{x\in V}
  \left|\inprod{g,x}\right| \right)^2 ,
\EquationName{gaussianproc}
\end{equation}
where $g\in\R^n$ has independent Gaussian
entries of mean $0$ and variance $1$. The bound on $m$ in \cite{KM05}
is actually
stated as $C\eps^{-2}(\gamma_2(V, \|\cdot\|_2))^2$ where $\gamma_2$ is the
$\gamma_2$ functional, but this is equivalent to
\Equation{gaussianproc} up to a constant factor; see
\cite{Talagrand05} for details. Note \Equation{gaussianproc}
easily implies the $m=O(d/\eps^2)$ bound for OSE's by letting $V$ be
the unit sphere in any $d$-dimensional subspace,  and also implies
$m=O(\delta_k^{-2}k\log(n/k))$ suffices for RIP matrices by letting
$V$ be the set of all $k$-sparse vectors of unit norm.

Note that the resolution of this question will not just be in terms of
the $\gamma_2$ functional. In particular, for constant $\delta_k$ we
see that $m,s=\Theta((\gamma_2(V))^2)$ is necessary and sufficient
when $V$ is the set of all unit norm $k$-sparse vectors. 
Even increasing $m$ to $\Theta((\gamma_2(V))^{2+\gamma})$ does not
decrease the lower bound on $s$ by much.
Meanwhile
for $V$ a unit sphere of a $d$-dimensional subspace, we can
simultaneously have $m = O((\gamma_2(V))^{2+\gamma}/\eps^2)$, and
$s=O(1/\eps)$ not depending on $\gamma_2(V)$ at all.

\bibliographystyle{alpha}

\bibliography{../allpapers}

\newcommand{\etalchar}[1]{$^{#1}$}
\begin{thebibliography}{DMIMW12}

\bibitem[AC09]{AC09}
Nir Ailon and Bernard Chazelle.
\newblock The {Fast} {Johnson--Lindenstrauss} transform and approximate nearest
  neighbors.
\newblock {\em SIAM J. Comput.}, 39(1):302--322, 2009.

\bibitem[Ach03]{Achlioptas03}
Dimitris Achlioptas.
\newblock Database-friendly random projections: {Johnson-Lindenstrauss} with
  binary coins.
\newblock {\em J. Comput. Syst. Sci.}, 66(4):671--687, 2003.

\bibitem[AHK06]{AHK06}
Sanjeev Arora, Elad Hazan, and Satyen Kale.
\newblock A fast random sampling algorithm for sparsifying matrices.
\newblock In {\em Proceedings of the 10th International Workshop on
  Randomization and Computation (RANDOM)}, pages 272--279, 2006.

\bibitem[AL09]{AL09}
Nir Ailon and Edo Liberty.
\newblock Fast dimension reduction using {Rademacher} series on dual {BCH}
  codes.
\newblock {\em Discrete Comput. Geom.}, 42(4):615--630, 2009.

\bibitem[AL11]{AL11}
Nir Ailon and Edo Liberty.
\newblock Almost optimal unrestricted fast {Johnson-Lindenstrauss} transform.
\newblock In {\em Proceedings of the 22nd Annual ACM-SIAM Symposium on Discrete
  Algorithms (SODA)}, pages 185--191, 2011.

\bibitem[Alo09]{Alon09}
Noga Alon.
\newblock Perturbed identity matrices have high rank: Proof and applications.
\newblock {\em Combinatorics, Probability {\&} Computing}, 18(1-2):3--15, 2009.

\bibitem[AV06]{AV06}
Rosa~I. Arriaga and Santosh Vempala.
\newblock An algorithmic theory of learning: Robust concepts and random
  projection.
\newblock {\em Machine Learning}, 63(2):161--182, 2006.

\bibitem[BD08]{BD08}
Thomas Blumensath and Mike~E. Davies.
\newblock Iterative hard thresholding for compressed sensing.
\newblock {\em J. Fourier Anal. Appl.}, 14:629--654, 2008.

\bibitem[BDDW08]{BDDW08}
Richard Baraniuk, Mark Davenport, Ronald DeVore, and Michael Wakin.
\newblock A simple proof of the restricted isometry property for random
  matrices.
\newblock {\em Constr. Approx.}, 28:253--263, 2008.

\bibitem[BI09]{BI09}
Radu Berinde and Piotr Indyk.
\newblock Sequential sparse matching pursuit.
\newblock In {\em Proceedings of the 47th Annual Allerton Conference on
  Communication, Control, and Computing}, pages 36--43, 2009.

\bibitem[BIPW10]{DIPW10}
Khanh~Do Ba, Piotr Indyk, Eric Price, and David~P. Woodruff.
\newblock Lower bounds for sparse recovery.
\newblock In {\em Proceedings of the 21st Annual ACM-SIAM Symposium on Discrete
  Algorithms (SODA)}, pages 1190--1197, 2010.

\bibitem[BIR08]{BIR08}
Radu Berinde, Piotr Indyk, and Milan Ru\v{z}ic.
\newblock Practical near-optimal sparse recovery in the {L1} norm.
\newblock In {\em Proceedings of the 46th Annual Allerton Conference on
  Communication, Control, and Computing}, pages 198--205, 2008.

\bibitem[BOR10]{BOR10}
Vladimir Braverman, Rafail Ostrovsky, and Yuval Rabani.
\newblock Rademacher chaos, random {Eulerian} graphs and the sparse
  {Johnson-Lindenstrauss} transform.
\newblock {\em CoRR}, abs/1011.2590, 2010.

\bibitem[Can08]{Candes08}
Emmanuel~J. Cand\`{e}s.
\newblock The restricted isometry property and its implications for compressed
  sensing.
\newblock {\em C. R. Acad. Sci. Paris}, 346:589--592, 2008.

\bibitem[Cha10]{ChandarThesis}
Venkat~B. Chandar.
\newblock {\em Sparse Graph Codes for Compression, Sensing, and Secrecy}.
\newblock PhD thesis, Massachusetts Institute of Technology, 2010.

\bibitem[CRT06a]{crt06b}
Emmanuel~J. Cand\`{e}s, Justin Romberg, and Terence Tao.
\newblock Robust uncertainty principles: Exact signal reconstruction from
  highly incomplete frequency information.
\newblock {\em IEEE Trans. Inf. Theory}, (52):489--509, 2006.

\bibitem[CRT06b]{crt06}
Emmanuel~J. Cand\`{e}s, Justin Romberg, and Terence Tao.
\newblock Stable signal recovery from incomplete and inaccurate measurements.
\newblock {\em Communications on Pure and Applied Mathematics}, 59(8), 2006.

\bibitem[CT05]{CT05}
Emmanuel~J. Cand{\`e}s and Terence Tao.
\newblock Decoding by linear programming.
\newblock {\em IEEE Trans. Inf. Theory}, 51(12):4203--4215, 2005.

\bibitem[CT06]{CT06}
Emmanuel~J. Cand\`{e}s and Terence Tao.
\newblock Near-optimal signal recovery from random projections: universal
  encoding strategies?
\newblock {\em IEEE Trans. Inf. Theory}, 52:5406--5425, 2006.

\bibitem[CW12]{CW12}
Kenneth~L. Clarkson and David~P. Woodruff.
\newblock Low rank approximation and regression in input sparsity time.
\newblock {\em CoRR}, abs/1207.6365v2, 2012.

\bibitem[DG03]{DG03}
Sanjoy Dasgupta and Anupam Gupta.
\newblock An elementary proof of a theorem of {Johnson} and {Lindenstrauss}.
\newblock {\em Random Struct. Algorithms}, 22(1):60--65, 2003.

\bibitem[DKS10]{DKS10}
Anirban Dasgupta, Ravi Kumar, and Tam{\'a}s Sarl{\'o}s.
\newblock A sparse {Johnson-Lindenstrauss} transform.
\newblock In {\em Proceedings of the 42nd ACM Symposium on Theory of Computing
  (STOC)}, pages 341--350, 2010.

\bibitem[DMIMW12]{DMMW12}
Petros Drineas, Malik Magdon-Ismail, Michael Mahoney, and David Woodruff.
\newblock Fast approximation of matrix coherence and statistical leverage.
\newblock In {\em Proceedings of the 29th International Conference on Machine
  Learning (ICML)}, 2012.

\bibitem[Don06]{don06}
David~L. Donoho.
\newblock Compressed sensing.
\newblock {\em IEEE Trans. Inf. Theory}, 52(4):1289--1306, 2006.

\bibitem[DTDlS12]{DTDS12}
David~L. Donoho, Yaakov Tsaig, Iddo Drori, and Jean luc Starck.
\newblock Sparse solution of underdetermined linear equations by stagewise
  orthogonal matching pursuit.
\newblock {\em IEEE Trans. Inf. Theory}, 58:1094--1121, 2012.

\bibitem[FM88]{FM88}
Peter Frankl and Hiroshi Maehara.
\newblock The {Johnson-Lindenstrauss} lemma and the sphericity of some graphs.
\newblock {\em J. Comb. Theory. Ser. B}, 44(3):355--362, 1988.

\bibitem[Fou11]{Foucart11}
Simon Foucart.
\newblock Hard thresholding pursuit: an algorithm for compressive sensing.
\newblock {\em SIAM J. Numer. Anal.}, 49(6):2543--2563, 2011.

\bibitem[GG84]{GG84}
Andrej~Y. Garnaev and Efim~D. Gluskin.
\newblock On the widths of the {Euclidean} ball.
\newblock {\em Soviet Mathematics Doklady}, 30:200--203, 1984.

\bibitem[GK09]{GK09}
Rahul Garg and Rohit Khandekar.
\newblock Gradient descent with sparsification: an iterative algorithm for
  sparse recovery with restricted isometry property.
\newblock In {\em Proceedings of the 26th Annual International Conference on
  Machine Learning (ICML)}, pages 337--344, 2009.

\bibitem[Gor88]{Gordon88}
Yehoram Gordon.
\newblock On {Milman's} inequality and random subspaces which escape through a
  mesh in $\mathbb{R}^n$.
\newblock {\em Geometric Aspects of Functional Analysis}, pages 84--106, 1988.

\bibitem[IM98]{IM98}
Piotr Indyk and Rajeev Motwani.
\newblock Approximate nearest neighbors: Towards removing the curse of
  dimensionality.
\newblock In {\em Proceedings of the 30th ACM Symposium on Theory of Computing
  (STOC)}, pages 604--613, 1998.

\bibitem[Ind01]{Indyk01}
Piotr Indyk.
\newblock Algorithmic applications of low-distortion geometric embeddings.
\newblock In {\em Proceedings of the 42nd Annual Symposium on Foundations of
  Computer Science (FOCS)}, pages 10--33, 2001.

\bibitem[IR08]{IR08}
Piotr Indyk and Milan Ru\v{z}ic.
\newblock Near-optimal sparse recovery in the {L1} norm.
\newblock In {\em Proceedings of the 49th Annual IEEE Symposium on Foundations
  of Computer Science (FOCS)}, pages 199--207, 2008.

\bibitem[JL84]{JL84}
William~B. Johnson and Joram Lindenstrauss.
\newblock Extensions of {Lipschitz} mappings into a {Hilbert} space.
\newblock {\em Contemporary Mathematics}, 26:189--206, 1984.

\bibitem[Ka{\v{s}}77]{Kashin77}
Boris~Sergeevich Ka{\v{s}}in.
\newblock The widths of certain finite-dimensional sets and classes of smooth
  functions.
\newblock {\em Izv. Akad. Nauk SSSR Ser. Mat.}, 41(2):334--351, 478, 1977.

\bibitem[KM05]{KM05}
Bo'az Klartag and Shahar Mendelson.
\newblock Empirical processes and random projections.
\newblock {\em J. Funct. Anal.}, 225(1):229--245, 2005.

\bibitem[KN10]{KN10}
Daniel~M. Kane and Jelani Nelson.
\newblock A derandomized sparse {Johnson-Lindenstrauss} transform.
\newblock {\em CoRR}, abs/1006.3585, 2010.

\bibitem[KN12]{KN12}
Daniel~M. Kane and Jelani Nelson.
\newblock Sparser {Johnson}-{Lindenstrauss} transforms.
\newblock In {\em SODA}, pages 1195--1206, 2012.

\bibitem[KW11]{KW11}
Felix Krahmer and Rachel Ward.
\newblock New and improved {J}ohnson-{L}indenstrauss embeddings via the
  {R}estricted {I}sometry {P}roperty.
\newblock {\em SIAM J. Math. Anal.}, 43(3):1269--1281, 2011.

\bibitem[LDP07]{LDP07}
Michael Lustig, David Donoho, and John~M. Pauly.
\newblock Sparse {MRI}: The application of compressed sensing for rapid {MR}
  {Imaging}.
\newblock {\em Magnetic Resonance in Medicine}, 58:1182--1195, 2007.

\bibitem[Mat08]{Matousek08}
Jir\'{\i} Matousek.
\newblock On variants of the {Johnson-Lindenstrauss} lemma.
\newblock {\em Random Struct. Algorithms}, 33(2):142--156, 2008.

\bibitem[MM12]{MM12}
Xiangrui Meng and Michael~W. Mahoney.
\newblock Low-distortion subspace embeddings in input-sparsity time and
  applications to robust linear regression.
\newblock {\em CoRR}, abs/1210.3135, 2012.

\bibitem[MP12]{MP12}
Gary~L. Miller and Richard Peng.
\newblock Iteratives approaches to row sampling.
\newblock Manuscript, 2012.

\bibitem[NN12]{NN12a}
Jelani Nelson and Huy~L. Nguy$\tilde{\hat{\mbox{e}}}$n.
\newblock {OSNAP}: Faster numerical linear algebra algorithms via sparser
  subspace embeddings.
\newblock Manuscript, 2012.

\bibitem[NT09]{NT09}
Deanna Needell and Joel~A. Tropp.
\newblock {CoSaMP}: Iterative signal recovery from incomplete and inaccurate
  samples.
\newblock {\em Appl. Comput. Harmon. Anal.}, 26:301--332, 2009.

\bibitem[NV09]{NV09}
Deanna Needell and Roman Vershynin.
\newblock Uniform uncertainty principle and signal recovery via regularized
  orthogonal matching pursuit.
\newblock {\em Foundations of Computational Mathematics}, 9(3):317--334, 2009.

\bibitem[NV10]{NV10}
Deanna Needell and Roman Vershynin.
\newblock Signal recovery from inaccurate and incomplete measurements via
  regularized orthogonal matching pursuit.
\newblock {\em IEEE Journal of Selected Topics in Signal Processing},
  4:310--316, 2010.

\bibitem[Sar06]{Sarlos06}
Tam{\'a}s Sarl{\'o}s.
\newblock Improved approximation algorithms for large matrices via random
  projections.
\newblock In {\em Proceedings of the 47th Annual IEEE Symposium on Foundations
  of Computer Science (FOCS)}, pages 143--152, 2006.

\bibitem[Tal05]{Talagrand05}
Michel Talagrand.
\newblock {\em The generic chaining: upper and lower bounds of stochastic
  processes}.
\newblock Springer Verlag, 2005.

\bibitem[TG07]{TG07}
Joel~A. Tropp and Anna~C. Gilbert.
\newblock Signal recovery from random measurements via orthogonal matching
  pursuit.
\newblock {\em IEEE Trans. Inf. Theory}, 53(12):4655--4666, 2007.

\bibitem[Tro11]{Tropp11}
Joel~A. Tropp.
\newblock Improved analysis of the subsampled randomized {Hadamard} transform.
\newblock {\em Adv. Adapt. Data Anal., Special Issue on Sparse Representation
  of Data and Images}, 3(1--2):115--126, 2011.

\bibitem[Vem04]{Vempala04}
Santosh Vempala.
\newblock {\em The random projection method}, volume~65 of {\em DIMACS Series
  in Discrete Mathematics and Theoretical Computer Science}.
\newblock American Mathematical Society, 2004.

\bibitem[WDL{\etalchar{+}}09]{WDLSA09}
Kilian~Q. Weinberger, Anirban Dasgupta, John Langford, Alexander~J. Smola, and
  Josh Attenberg.
\newblock Feature hashing for large scale multitask learning.
\newblock In {\em Proceedings of the 26th Annual International Conference on
  Machine Learning (ICML)}, pages 1113--1120, 2009.

\bibitem[ZWSP08]{ZWSP08}
Yunhong Zhou, Dennis~M. Wilkinson, Robert Schreiber, and Rong Pan.
\newblock Large-scale parallel collaborative filtering for the netflix prize.
\newblock In {\em Proceedings of the 4th International Conference on
  Algorithmic Aspects in Information and Management (AAIM)}, pages 337--348,
  2008.

\end{thebibliography}

\end{document}